\documentclass[a4paper, 11pt]{article}
\usepackage[margin=1in]{geometry}

\usepackage{xspace}
\usepackage[strings]{underscore} 
\usepackage{cite} 
\usepackage{hyperref}

\usepackage[draft,mode=multiuser]{fixme}
\usepackage{microtype}

\usepackage{amssymb}
\usepackage{amsmath,amsthm}

\usepackage{authblk}
\newcommand{\email}[1]{{#1}}

\title{Block Palindromes: \\ A New Generalization of Palindromes}

\author[1]{Keisuke~Goto}
\author[2]{Tomohiro~I}
\author[3]{Hideo~Bannai}
\author[3]{Shunsuke~Inenaga}

\affil[1]{Fujitsu Laboratories Ltd., Kawasaki, Japan.
\email{goto.keisuke@jp.fujitsu.com}}
\affil[2]{Kyushu Institute of Technology, Iizuka, Japan.
\email{tomohiro@ai.kyutech.ac.jp}}
\affil[3]{Kyushu University, Fukuoka, Japan.
\email{\{bannai, inenaga\}@inf.kyushu.ac.jp}}

\newtheorem{theorem}{Theorem} 

\newtheorem{lemma}{Lemma}

\providecommand{\keywords}[1]{\linebreak \linebreak \textbf{\textit{Keywords:}} #1}

\newtheorem{remarkx}{Remark}
\newtheorem{examplex}{Example}

\makeatletter
\renewenvironment{proof}[1][\proofname]{\par
  \normalfont
  \topsep6\p@\@plus6\p@ \trivlist
  \item[\hskip\labelsep{\bfseries #1}\@addpunct{\bfseries.}]\ignorespaces
}{%
  \endtrivlist
}
\renewcommand{\proofname}{Proof}
\makeatother

\fxsetup{theme=color}
\fxuseenvlayout{colorsig}
\fxusetargetlayout{color}
\FXRegisterAuthor{kg}{akg}{KG}
\FXRegisterAuthor{hb}{ahb}{HB}
\FXRegisterAuthor{ti}{ati}{TI}
\definecolor{fxnote}{rgb}{1.0,0,0}    
\definecolor{fxtarget}{rgb}{1.0, 0, 0}
\fxsetface{target}{\bfseries}
\fxsetface{margin}{\bfseries \scriptsize}

\usepackage[whole]{bxcjkjatype}

\newcommand{\T}{T}

\newcommand{\bb}{b}
\newcommand{\cc}{c}
\newcommand{\dd}{d}
\newcommand{\ee}{e}
\newcommand{\ff}{f}

\newcommand{\ii}{i}
\newcommand{\jj}{j}
\newcommand{\kk}{k}
\newcommand{\mm}{m}
\newcommand{\nn}{n}

\newcommand{\rr}{r}
\newcommand{\ww}{w}
\newcommand{\xx}{x}
\newcommand{\yy}{y}
\newcommand{\zz}{z}

\newcommand{\FF}{F}

\newcommand{\LL}{L}

\newcommand{\NN}{N}

\newcommand{\PP}{P}

\newcommand{\vg}{g}

\newcommand{\vA}{A}
\newcommand{\vB}{B}

\newcommand{\floor}[1]{\lfloor #1 \rfloor}
\newcommand{\ceil}[1]{\lceil #1 \rceil}
\newcommand{\LCE}{\mathit{LCE}}

\newcommand{\MBP}{\mathit{MBP}}

\date{}
\begin{document}

\maketitle

\begin{abstract}
We study a new generalization of palindromes and gapped palindromes called \emph{block palindromes}.
A block palindrome is a string that becomes a palindrome when identical substrings are replaced with a distinct character.
We investigate several properties of block palindromes and
in particular, study substrings of a string which are block palindromes.
In so doing, we introduce the notion of a \emph{maximal block palindrome},
which leads to a compact representation of all block palindromes that occur in a string.
We also propose an algorithm which enumerates all maximal block palindromes that appear in a given string $\T$ in $O(|\T| + \|\MBP(\T)\|)$ time, where $\|\MBP(\T)\|$ is the output size,
which is optimal unless all the maximal block palindromes
can be represented in a more compact way.
\keywords{palindrome, enumeration algorithm, factorization}
\end{abstract}

\section{Introduction}
A palindrome is a string that is equal to its reverse, e.g., ``\texttt{Able_was_I_ere_I_saw_Elba}'' (we treat upper and lower characters are the same for simple explanations).
Palindromes have been studied in combinatorics on words and stringology.

Many research focused on finding palindromic structure of a string.
Manacher~\cite{Manacher75} proposed a beautiful algorithm that enumerates all maximal palindromes of a string.
Kosolobov et al.~\cite{KosolobovRS15} proved that, a language $\PP^\kk$ can be recognizable in $O(\kk \NN)$ time, where $\PP$ is the language of all nonempty palindromes and $\NN$ is the length of an input string.
Alatabbi et al.~\cite{Alatabbi13} considered maximal palindromic factorization in which all factors are maximal palindromes.
They also consider a problem of computing the fewest palindromic factorization, and proposed off-line linear-time algorithms.
Later, I et al.~\cite{ISIBT14} and Fici et al.~\cite{FiciGKK14} independently proposed on-line $O(\NN \log \NN)$-time algorithms, where $\NN$ is the length of an input string.
Similar problems were also considered, such as, computing palindromic length~\cite{BorozdinKRS17}, computing palindromic covers~\cite{ISIBT14}, computing palindromic pattern matching~\cite{I2013}.

A gapped palindrome is a generalization of a palindrome that becomes a palindrome when a center substring is replaced by a character, where the center substring is a substring whose beginning and ending positions are equally far from the beginning and ending positions of the input string, respectively.
For example, ``\texttt{Madam,_he_is_Adam}'' is a gapped palindrome, and it becomes a palindrome if the center substring ``\texttt{m,_he_is_}'' is replaced by a character.
Gapped palindromes play an important role in molecular biology since they model a hairpin data structure of DNA and RNA sequences, see e.g.~\cite{Gerald08}.
Several problems were considered such as, enumeration of exact gapped palindromes of a string~\cite{KolpakovK09} and also enumeration of approximate gapped palindromes~\cite{NarisadaDNIS17,HsuCC10}, finding maximal length of long armed or and constrained length gapped palindrome~\cite{Gupta2016}.

In this paper, we consider the notion of block palindromes~\cite{BIO2015}, which is
a new generalization of palindromes and also gapped palindromes~\footnote{Block palindromes were firstly introduced in a problem of 2015 British Informatics Olympiad~\cite{BIO2015}, but we did not know the existence at the first version of this paper.}.
A block palindrome is a string that becomes a palindrome when identical substrings are replaced with a distinct character.
More precisely, a block palindrome is a ``symmetric'' factorization $\ff=\ff_{-\nn} \cdots \ff_{-1} \ff_0 \ff_1 \cdots \ff_\nn$ of a string with the center factor $\ff_0$ is a string (which may be empty) and each of other factor $\ff_{-\ii}=\ff_{\ii}$ is a non-empty string for any $1 \le \ii \le \nn$.
We also call a factor a block.
For convenience, let $\ff=\ff_0$ when $\nn=0$.
For example, a factorization ``\texttt{To|kyo|_|and|_|Kyo|to}'' is a block palindrome, where ``\texttt{|}'' is a mark to distinguish adjacent blocks.
Palindromes and gapped palindromes are special cases of block palindromes:
For a palindrome, all blocks are characters, and for a gapped palindrome, the center block $\ff_0$ is a string and the other blocks are characters.

We investigate several properties of block palindromes.
We introduce the notion of maximal block palindromes to concisely represent all block palindromes in a string,
and propose an algorithm which enumerates all maximal block palindromes in a string $\T$ in $O(|\T| + \|\MBP(\T)\|)$ time, where $\|\MBP(\T)\|$ is the output size. This is optimal unless all the maximal block palindromes can be represented in a more compact way.

\section{Preliminaries}

Let $\Sigma$ be an integer alphabet.
An element of $\Sigma^*$ is called a \emph{string}.
The string of length 0 is called the \emph{empty} string,
and is denoted by $\varepsilon$.
Although $\varepsilon$ is not contained in $\Sigma$,
we sometimes call $\varepsilon$ the empty character for convenience.
For a string $\T=\xx \yy \zz$, $\xx$, $\yy$ and $\zz$ are called a \emph{prefix}, \emph{substring}, and \emph{suffix} of $\T$, respectively.
In particular, a prefix (resp. suffix) $\xx$ of $\T$ is called a \emph{proper} prefix (resp. suffix) iff $\xx \neq \T$.
A non-empty string that is a proper prefix and also a proper suffix of $\T$ is called a \emph{border} of $\T$.
Hence, a string of length $\NN$ can have
at most $\NN-1$ borders of length ranging from 1 to $\NN-1$.
A string which does not have any borders is called an \emph{unbordered} string.
For $1 \leq \ii \leq \jj \leq |\T|$,
a substring of $\T$ which begins at position $\ii$ and ends at
position $\jj$ is denoted by $\T[\ii \ldots \jj]$.
For convenience, let $\T[\ii \ldots \jj] = \varepsilon$ if $\jj < \ii$.

In this paper, we also consider \emph{half-positions} $\kk+1/2$ for integers $0 \leq \kk \leq |\T|$.
For convenience, for a half-position $\ii$ and an integer $\rr$ such that $1/2 \leq \ii-\rr \leq \ii + \rr \leq |\T|+1/2$, let $\T[\ii-\rr \ldots \ii+\rr] = \T[\ceil{\ii-\rr} \ldots \floor{\ii+\rr}]$.
Note that $\T[\ii]$ for a half-position $\ii$ is the empty character.
The position $\cc=(|\T|+1)/2$ is called the \emph{center position} of $\T$,  $\T[\cc]$ is called the \emph{center character} of $\T$, and $\T[\cc -\dd \ldots \cc+\dd]$ for an integer $\dd$ is  called a \emph{center substring} of $\T$.

For a string $\T$ and integers $1 \leq \ii, \jj \leq |\T|$, a \emph{longest common extension} (LCE) query $\LCE_\T(\ii, \jj)$ asks the length of the longest common prefix of the two suffixes $\T[\ii \ldots |\T|]$ and $\T[\jj \ldots |\T|]$ of $\T$.
When clear from the context, $\LCE_\T(\ii, \jj)$ is abbreviated as $\LCE(\ii, \jj)$.
It is well known that if $\T$ is drawn from an integer alphabet of size polynomially bounded in $|\T|$, then 
LCE queries for $\T$ can be answered in constant time
after an $O(|\T|)$-time preprocessing, e.g.,
by constructing the suffix tree of $\T$ and a data structure for lowest common ancestor queries on the tree~\cite{Gusfield1997AST}.

For a block palindrome $\ff=\ff_{-\nn} \cdots \ff_{-1} \ff_0 \ff_1 \cdots \ff_\nn$, the length of $\ff$ denoted by $|\ff|$ is the total length of blocks, and the size of $\ff$ denoted by $\|\ff\|$ is the number of non-empty blocks.
A block palindrome is \emph{even} if its size is even
(that is, the center block $\ff_0$ is the empty string),
and otherwise \emph{odd} (that is, the center block $\ff_0$ is non-empty).

\section{Properties of Block Palindromes}
In this section, we investigate the properties of block palindromes.
We assume that $\T$ is an input string of length $\NN$ in the rest of the paper.

Since there are $O(2^\NN)$ factorization of $\T$ and block palindromes are symmetric, there are $O(2^{\NN/2})$ block palindromes of $\T$.
Moreover, there is a tight example that $\T$ consists of only the same characters.

Although there are a huge number of block palindromes of a string, they are very redundant.
To look for more essential properties of block palindromes, we define
the \emph{largest block palindrome} which is a representative of other block palindromes.
A block palindrome $\ff=\ff_{-\nn} \cdots \ff_{\nn}$ of $\T$ that has
the largest number of blocks among all block palindromes of $\T$ is called the largest block palindrome.
Note that each block $\ff_\ii$ for $0 \leq \ii \leq \nn$ is an unbordered substring and $\ff_\ii$ for $0 < \ii \leq \nn$ is  the shortest border of $\T[\kk + 1 \ldots \NN - \kk]$, where $\kk=0$ if $\ii=\nn$ and $\kk=|\ff_{\ii+1} \cdots \ff_\nn|$ otherwise.
So, the largest block palindrome of $\T$ is unique.
The largest block palindrome is a representative of all block palindromes in the sense that all block palindromes can be represented as block palindromes of $\ff$.

A natural and prompt question would be about how to efficiently compute the largest block palindrome of $\T$.
The following theorem answers this question.

\begin{theorem}\label{theorem:largest}
	The largest block palindrome $\ff_{-\nn} \cdots \ff_{\nn}$ of $\T$ can be computed in $O(\NN)$ time.
\end{theorem}
\begin{proof}
	We construct a data structure in $O(\NN)$ time that can answer any LCE query in constant time.

	We greedily compute the blocks from outside $\ff_{\nn}$ to inner $\ff_{1}$ by LCE queries.
	We assume that we compute the shortest border $\ff_{\ii}$ of $\T[\bb \ldots \ee]$.
	For $\kk=1$ to $\floor{(\ee - \bb + 1)/2}$, we check whether $\T[\bb \ldots \bb + \kk - 1]$ is the border of $\T[\bb \ldots \ee]$ or not by checking whether $\LCE(\bb, \ee-\kk+1) \geq \kk$ or not.
	If $\T[\bb \ldots \ee]$ does not have any border, we obtain $\ff_0 = \T[\bb \ldots \ee]$.
	Otherwise, we obtain the shortest border $\ff_\ii=\T[\bb \ldots \bb + \kk-1]$ of $\T[\bb \ldots \ee]$, and compute the more inner blocks for $\T[\bb + \kk \ldots \ee - \kk]$.
	Since the number of LCE queries is $O(\NN)$ and each LCE query takes constant time, the largest block palindrome of $\T$ can be computed in $O(\NN)$ time.
\qed
\end{proof}

So far, we have considered only block palindromes that are equal to $\T$ itself.
Next, we consider block palindromes that appear as substrings in $\T$.
We define a \emph{maximal block palindrome} which is a representative of some block palindromes in $\T$, and study how many maximal block palindromes can appear in $\T$.

For a half-position $1 \leq \cc \leq \NN$ and an integer $1 \leq \dd \leq \NN/2$, let $\FF_\T(\cc, \dd)=\{\ff | \ff=\ff_{-\nn} \cdots \ff_0 \cdots \ff_{\nn} \mbox{ is the largest block palindrome}, \ff_0=\T[\cc - \dd + 1 \ldots \cc + \dd-1], \ff=\T[\cc - \dd - \kk+1 \ldots \cc + \dd + \kk-1], \kk = |\ff_1 \cdots \ff_\nn| \}$ be the set of largest block palindromes whose center positions are the same and whose center blocks appear at $\T[\cc-\dd+1 \ldots \cc+\dd-1]$.
When context is clear, we denote $\FF_{\T}$ by $\FF$.
For a string $\T$, a largest block palindrome $\ff \in \FF(\cc, \dd)$ such that $|\ff|$ is the longest, namely the number of blocks are maximal among all largest block palindromes of $\FF(\cc, \dd)$, is called a \textit{maximal block palindrome}.

We remark that the maximal block palindrome of $\FF(\cc, \dd)$ is a representative of all the largest block palindromes of $\FF(\cc, \dd)$.

\begin{remarkx}
	\label{lem:maximal-palindrome-is-representative}
  For a half-position $1 \leq \cc \leq \NN$ and an integer $1 \leq \dd \leq \NN/2$, any largest block palindrome $\ff=\ff_{-\nn} \cdots \ff_{\nn} \in \FF(\cc, \dd)$ is a sub-factorization of the maximal block palindrome $\vg = \vg_{-\mm} \cdots \vg_{\mm} \in \FF(\cc, \dd)$, that is, $\nn \leq \mm$ and $\ff_{\ii}=\vg_{\ii}$ for $0 \leq \ii \leq \nn$.
\end{remarkx}
\begin{proof}
	We assume that the statement does not hold.
	Let $\ff_\jj$ be a block that $\ff_{\jj} \neq \vg_{\jj}$, and $\jj=0$ or $\ff_{\ii}=\vg_{\ii}$ for $0 \leq \ii < \jj \leq \nn$.
	If $|\ff_{\jj}| < |\vg_{\jj}|$, $\ff_{\jj}$ is a border of $\vg_{\jj}$ and it contradicts that $\vg_{\jj}$ is the largest block palindrome.
	We can say the same things for the case $|\ff_{\jj}| > |\vg_{\jj}|$.
	Therefore, such $\ff_\jj$ and $\vg_\jj$ do not exist and this statement holds.
	\qed
\end{proof}

We are interested in how many maximal block palindromes can appear in $\T$.
It is trivially upper bounded by $O(\NN^2)$ since there are $O(\NN^2)$ substrings which can be center substrings.
If there is no limitation on the size of maximal block palindromes, we can easily see that it is tight.
For a string $\T$ of length $\NN$ in which the characters are all distinct, any substring $\ww$ is unbordered, and there is at least one maximal block palindrome that contains $\ww$ as a center block.
Thus, $\T$ can contain $\Theta(\NN^2)$ maximal block palindromes.
The following example says that the number of maximal block palindromes having three blocks has also the same tight upper bound.

\begin{examplex}
	\label{lem:bound-of-num-maximal-palindromes}
	The number of maximal block palindromes in $\T=\mathtt{a}^n\mathtt{b}^n\mathtt{a}\mathtt{b}\mathtt{a}^n\mathtt{b}^n$ that have at least three blocks is $\Theta(\NN^2)$, where $\cc^\nn$ for a character $\cc$ denotes run of $\cc$ of length $\nn$ , and $\nn=(\NN-2)/4$.
\end{examplex}
For convenience, we denote $\T$ by $\T=\vA_0 \vB_1 \vA_1 \vB_2 \vA_2 \vB_3$, where $\vA_0$, $\vB_1$, $\vA_1$, $\vB_2$, $\vA_2$, and $\vB_3$ are strings $\mathtt{a}^n$, $\mathtt{b}^n$, $\mathtt{a}$, $\mathtt{b}$, $\mathtt{a}^n$, and $\mathtt{b}^n$, respectively.
There are maximal block palindromes of size three that, for $1<\ii \leq \nn$, $1<\jj \leq \nn$, $\T[\nn-\jj+1 \ldots \NN-\nn+\ii-1]$=$(\vA_0[\nn-\jj+1 \ldots \nn] \vB_1[1..\ii-1])(\vB_1[\ii \ldots \nn]\vA_1 \vB_2 \vA_2[1 \ldots \jj])(\vA_2[\nn-\jj+1 \ldots \nn]\vB_3[1 \ldots \ii-1])$ and they are unbordered, where the parentheses indicate blocks.


Remark that the upper bound is reduced to $O(\NN)$ if we impose a limitation on the lengths of center blocks.
\begin{remarkx}
	\label{lem:limited-center-block}
  For any constant $\kk \ge 0$, a string of length $\NN$ can contain $\Theta(\NN)$
  maximal block palindromes whose center blocks are of length $\le \kk$ because there are $O(\NN)$ possible center blocks.
  In particular, a string contains at most $\NN - 1$ maximal block palindromes of even size (i.e., the center blocks must be empty)
  because the number of occurrences of center blocks are at most $\NN-1$.
\end{remarkx}

The following lemma shows an interesting property of maximal block palindromes, and this property can be used for the proof of Lemma~\ref{lem:size-maximal-block-palindrome}.
\begin{lemma}
	\label{lem:there-is-no-same-starting-positions-of-factors}
  For a half-position $1 \leq \cc \leq \NN$ and two integers $1 \leq \dd<\dd^\prime \leq \NN/2$, two largest block palindromes $\ff=\ff_{-\nn} \cdots \ff_{\nn} \in \FF(\cc, \dd)$ and $\vg = \vg_{-\mm} \cdots \vg_{\mm} \in \FF(\cc, \dd^\prime)$ do not share the block boundaries, namely, the ending positions of blocks $\kk_\ii$ and $\kk^\prime_\ii$ such that $\kk_\ii=\cc + \dd -1 + |\ff_1 \cdots \ff_\ii|$ and $\kk^\prime_\ii=\cc + \dd^\prime -1 + |\vg_1 \cdots \vg_\jj|$ do not equal for any $1 \leq \ii \leq \nn$ and $1 \leq \jj \leq \mm$.
\end{lemma}
\begin{proof}
  Similar to Remark~\ref{lem:maximal-palindrome-is-representative}, if we assume that this lemma does not hold, a block of $\ff$ or $\vg$ must have a border and it contradicts that $\ff$ and $\vg$ are the largest block palindromes.
	\qed
\end{proof}

Let $\|\MBP(\T)\|$ denote the sum of the sizes of all maximal block palindromes in $\T$.
\begin{lemma}
  \label{lem:size-maximal-block-palindrome}
  For any string $\T$ of length $\NN$, $\|\MBP(\T)\| \le \NN(2\NN-1)$.
\end{lemma}
\begin{proof}
  From Lemma~\ref{lem:there-is-no-same-starting-positions-of-factors}, any two largest block palindromes, whose center positions
  are same but center blocks are different, do not share the block boundaries.
  This implies that, for a half-position $\cc$, the number of blocks of maximal block palindromes whose center position is $\cc$ is up to $\NN$.
  Since there are $2\NN-1$ center positions, we have $\|\MBP(\T)\| \le \NN(2\NN-1)$.
	\qed
\end{proof}

\section{Enumeration of Maximal Block Palindromes}
In this section, we consider how to enumerate all the maximal block palindromes $\MBP(\T)$.
A brute-force approach based on Theorem~\ref{theorem:largest} would compute the largest block for every possible substring $\T[\bb \ldots \bb + \ell - 1]$
(while suppressing output of non-maximal ones), which takes $\Theta(\sum_{\ell = 1}^{\NN} \ell (\NN - \ell)) = \Theta(\NN^3)$ time.

We propose an optimal solution running in $o(\NN^3)$ time.
\begin{theorem}
  All maximal block palindromes that appear in $\T$ can be enumerated in $O(\NN + \|\MBP(\T)\|)$ time, 
  where $\|\MBP(\T)\|$ is the output size.
\end{theorem}

We actually consider a variant of the problem:
We propose an algorithm to enumerate all the maximal block palindromes of size $\ge 2$,
whose total output size is denoted by $\|\MBP_{\ge 2}(\T)\|$, in optimal $O(\NN + \|\MBP_{\ge 2}(\T)\|)$ time.
That is to say, we can completely ignore maximal block palindromes of size $1$,
which might not be interesting if we focus on palindromic structures in $\T$.
If we want to enumerate $\MBP(\T)$, we can do that by slightly modifying the algorithm.

Our algorithm proceeds in two steps:
(i) enumerate all the pairing unbordered blocks for all center positions in a batch processing, and
(ii) build maximal block palindromes from the enumerated blocks.

In Step (i), we firstly enumerate every pair of occurrences of an unbordered substring in $\T$.
Note that the pair will be a component of a maximal block palindrome, and the total number of enumerated pairs is $O(\|\MBP_{\ge 2}(\T)\|)$.
We preprocess $\T$ in $O(\NN)$ time and space to support LCE queries in constant time.
We also compute, for every character in $\T$, the list storing all the occurrences of the character in increasing order,
all of which can be obtained by sorting the positions $\ii$ of $\T$ with the key $\T[\ii]$ by radix sort in $O(\NN)$ time and space.

Now we focus on an occurrence $\bb$ of $\T[\bb]$,
and identify every pair of occurrences of an unbordered substring such that the left one starts at $\bb$.
Let $\bb < \bb_1 < \bb_2 < \cdots < \bb_k$ be the occurrences of $\T[\bb]$ in $\T[\bb \ldots \NN]$.
We process $\bb_i \in \{\bb_1, \ldots, \bb_k\}$ in increasing order to identify common unbordered substrings
starting at $\bb$ and $\bb_i$ using $\LCE$ queries.
At the first round for $\bb_1$, we see that for any $\ell$ with $1 \le \ell \le \min (\LCE(\bb, \bb_1), \bb_1 - \bb)$,
the common substring of length $\ell$ starting at $\bb$ and $\bb_1$ is unbordered,
and thus, we report each pair of such unbordered substrings.
While processing $\bb_i \in \{\bb_1, \ldots, \bb_k\}$ in increasing order,
we maintain a set $\LL$ of positive integers $\ell$ (by a sorted list of intervals) such that $\T[\bb \ldots \bb + \ell - 1]$ has a border 
caused by the common substrings starting at $\bb$ and $\bb_i$'s processed so far.
We use $\LL$ to efficiently skip $\ell$'s such that $\T[\bb \ldots \bb + \ell - 1]$ has a border in the later rounds.
For example, in the first round, we add the interval $[\bb_1 - \bb + 1 \ldots \bb_1 - \bb + \LCE(\bb, \bb_1)]$ to $\LL$ (which is initially empty)
as, for any $\ell \in [\bb_1 - \bb + 1 \ldots \bb_1 - \bb + \LCE(\bb, \bb_1)]$, $\T[\bb \ldots \bb + \ell - 1]$ has a border 
caused by the common substring starting at $\bb$ and $\bb_1$.
When processing $\bb_i$ for $1 < i \le k$, we see that for any $\ell \in [1 \ldots \min (\LCE(\bb, \bb_i), \bb_i - \bb)] \setminus \LL$,
the common substring of length $\ell$ starting at $\bb$ and $\bb_i$ is unbordered.
Updating $\LL$ can be easily done in $O(1)$ time by adding (merging if necessary) the interval $[\bb_i - \bb + 1 \ldots \bb_i - \bb + \LCE(\bb, \bb_i)]$
to $\LL$ (observe that the new interval is always pushed back to $\LL$ or merged with the last interval of $\LL$
as we process $\{\bb_1, \ldots, \bb_k\}$ in increasing order).
Note that $[1 \ldots \min (\LCE(\bb, \bb_i), \bb_i - \bb)] \setminus \LL$ always contains $1$,
and we can incrementally enumerate its element in constant time per element
because $\LL$ is maintained as a sorted list of intervals.
Thus, the computation cost can be charged to the number of output, i.e., it runs in $O(\NN + \|\MBP_{\ge 2}(\T)\|)$ time in total.

When we find a pair of occurrences $\bb_{l} < \bb_{r}$ of an unbordered substring of length $\ell$,
we list it up as a triple $(\cc, \bb_{r}, \bb_{r} + \ell)$, where $\cc = (\bb_{l} + \bb_{r} + \ell - 1) / 2$ is the center of the pairing blocks.
After listing up all those triples, we sort them using the first and second elements as keys by radix sort,
which can be done in $O(\NN + \|\MBP_{\ge 2}(\T)\|)$ time and space.

Now we are ready to proceed to Step (ii) in which we build the maximal block palindromes from the sorted list of triples computed in Step (i).
For building the maximal block palindromes with center $\cc$,
we scan the sublist of triples having center $\cc$ and connect the pairing blocks whose beginning and ending positions are adjacent
using the information of the second (the beginning position of the block) and third (the ending position of the block plus one) elements of the triples.
We build all the $\cc$-centered maximal block palindromes by extending their blocks outwards simultaneously
with a $0$-initialized array $A$ of length $\NN$.
When we look at a triple $(\cc, \bb_{r}, \bb_{r} + \ell)$,
we write $\bb_{r}$ to $A[\bb_{r} + \ell]$, and connect the block with the block ending at $\bb_{r} - 1$ if such exists
(which can be noticed by the information $A[\bb_{r}] \neq 0$).
Since the block boundaries are not shared due to Lemma~\ref{lem:there-is-no-same-starting-positions-of-factors},
the information written in $A$ can be propagated correctly to extend the blocks.
It runs in time linear to the size of the sublist.
We can also clear $A$ in the same time by scanning the sublist again while writing $0$ to the entries we touched.

Since the initialization cost $O(\NN)$ of $A$ is payed once in the very beginning of Step (ii)
and the other computation cost can be charged to the output size, the total time complexity is $O(\NN + \|\MBP_{\ge 2}(\T)\|)$.

For enumerating $\MBP(\T)$, we modify Step (ii).
While scanning the sublist for center $\cc$,
we can identify all the positions $\ee \ge \cc$ such that $\ee$ is not an ending position of some pairing block,
for which the substring $\T[2 \cc - \ee \ldots \ee]$ is unbordered.
If the unbordered substring cannot be extended outwards by blocks (which can also be checked while scanning the sublist),
it is the maximal block palindrome of size $1$ to output for $\MBP(\T)$.
The algorithm runs in $O(\NN + \|\MBP(\T)\|)$ time in total as the additional cost can be charged to the output size.

\section{Conclusions}
In this paper, we investigated several properties of block palindromes which are the generalization of palindromes and gapped palindromes.
We also proposed an optimal-algorithm to enumerate all maximal block palindromes appearing in a given string.
As mentioned in Remark~\ref{lem:limited-center-block},
if we impose a limitation on the lengths of center blocks,
the upper bound of the number of maximal block palindromes is reduced to $O(\NN)$,
where $\NN$ is the length of an input string.
In particular, for maximal block palindromes of even size, the center blocks are super restricted to be empty.
The situation is similar to considering ordinal palindromes (in which the center blocks are strict)
versus maximal gapped palindromes (in which the restriction on the center blocks are relaxed).
It would be interesting to investigate the properties of maximal block palindromes whose center blocks have restricted lengths
and develop efficient algorithms to enumerate only such a subset of maximal block palindromes.

\bibliographystyle{plainurl}
\bibliography{ref}

\end{document}